\documentclass[letter,11pt]{article}
\usepackage{graphicx}
\usepackage{subcaption}
\usepackage{verbatim}
\usepackage{titling}
\usepackage{cancel}
\usepackage{enumitem} 
\setlist{noitemsep,topsep=0pt,parsep=0pt} 
\usepackage[margin=1cm]{caption} 
\usepackage{mathtools} 
\usepackage{multirow}
\usepackage[all]{xy}
\usepackage{tikz}

\usetikzlibrary{arrows,backgrounds,calc,fit,decorations.pathreplacing,decorations.markings,shapes.geometric}

\tikzset{every fit/.append style=text badly centered}

\usepackage{tyson}

\usepackage[bookmarks=true,hypertexnames=false]{hyperref}
\hypersetup{colorlinks=true, citecolor=blue, linkcolor=red, urlcolor=blue}
\makeatletter

\newcommand{\Rmnum}[1]{\expandafter\@slowromancap\romannumeral #1@}
\makeatother

\newcommand{\Holant}{\operatorname{Holant}}

\newcommand{\holant}[2]{\ensuremath{\Holant\left(#1\mid #2\right)}}

\newcommand{\numP}{{\rm \#P}}

\newenvironment{remark}{\medskip{\bfseries \noindent Remark:}}{\par\medskip}{\par\medskip}

\tikzstyle{internal} = [draw, fill, shape=circle]
\tikzstyle{external} = [shape=circle]
\tikzstyle{square}   = [draw, fill, rectangle]
\tikzstyle{triangle} = [draw, fill, regular polygon, regular polygon sides=3, inner sep=3pt]
\tikzstyle{pentagon} = [draw, fill, regular polygon, regular polygon sides=5, inner sep=2pt, minimum size=14pt]

\begin{document}

\title{{\bf Dichotomy Result on 3-Regular Bipartite Non-negative Functions}}

\vspace{0.3in}

	\title{\bf Dichotomy Result on 3-Regular Bipartite Non-negative Functions}
		\vspace{0.3in}
		\author{Austen Z. Fan\thanks{Department of Computer Sciences, University of Wisconsin-Madison.}\\ {\tt zfan64@wisc.edu}
\and  Jin-Yi Cai\thanks{Department of Computer Sciences, University of Wisconsin-Madison. Supported by NSF CCF-1714275.
			} \\ {\tt jyc@cs.wisc.edu}
		}
	
\date{}
\maketitle

\begin{abstract}
We prove a complexity dichotomy theorem for a class of Holant
problems on 3-regular bipartite graphs. 
Given an arbitrary  nonnegative weighted
symmetric constraint function $f = [x_0, x_1, x_2, x_3]$,
we prove that the bipartite Holant problem
$\holant{f}{(=_3)}$
is \emph{either} computable in polynomial time \emph{or}
\#P-hard. The dichotomy criterion on $f$ is explicit.
 \end{abstract}

\section{Introduction}
Holant problems are also called edge-coloring models.
They can express a broad class of counting problems,
such as counting matchings ({\sc \#Matchings}),
perfect matchings ({\sc \#PM}), edge-colorings, cycle coverings,
and a host of counting orientation problems such as counting
Eulerian orientations.

Given an input graph $G = (V, E)$, we identify each 
edge $e \in E$ as a
variable over some finite domain $D$, and identify
each vertex $v$ as a constraint function $f_v$.
Then the \emph{partition function} is the following sum of
product $\sum_{\sigma: E \rightarrow D}
\prod_{v \in V} f_v(\sigma |_{E(v)})$, where
$E(v)$ denotes the edges incident to $v$
and $f_v(\sigma |_{E(v)})$ is the evaluation
of $f_v$ on the restriction of $\sigma$ on $E(v)$.
For example, {\sc \#Matchings} and {\sc \#PM}  are
counting  problems
specified by the constraint function
{\sc At-Most-One}, respectively {\sc Exact-One},
which outputs value 1 if
the input bits have \emph{at most} one 1, respectively \emph{exactly} one 1,
and outputs 0 otherwise. Thus,
every term $\prod_{v \in V} f_v(\sigma |_{E(v)})$ evaluates
to 1 or 0, and is 1 iff the assignment $\sigma$ is a
mathching,  respectively a perfect mathching.
The framework
of Holant problems  is intimately related to Valiant's
holographic algorithms~\cite{valiant2008holographic}. 
Holant problems can encompass
all counting constraint satisfaction problems (\#CSP),
in particular all vertex-coloring models from statistical
physics.
For \#CSP, Bulatov~\cite{Bulatov} proved a 
sweeping complexity dichotomy. Dyer and Richerby~\cite{DyerRicherby} gave another proof of this dichotomy and also showed that the dichotomy is decidable.  Cai, Chen, and Lu in \cite{CCL} 
extended this to nonnegative weighted case. This was further generalized to complex weighted \#CSP~\cite{CaiChen-JACM}.
However, no full dichotomy has been proved for Holant problems.

Every \#CSP problem can be easily expressed
as a Holant problem. On the other hand,
Freedman, Lov\'asz and Schrijver~\cite{Freedman-Lovasz-Schrijver-2007} proved that the prototypical Holant problem
\#PM cannot be expressed as a vertex-coloring model
using any real constraint function, and this was extended
to complex constraint functions~\cite{cai-govorov-express}. Thus Holant problems
are strictly more expressive.

In this paper we consider the Boolean domain
$D = \{0,1\}$. Significant knowledge has been
gained about the complexity of Holant problems~\cite{DBLP:conf/icalp/Backens18,DBLP:journals/talg/BackensG20,cai2017complexity,cai2016complete,lin2018complexity,shao2020dichotomy}.
However, there has been very limited progress
on bipartite Holant problems where the input
graph $G = (U,V,E)$ is bipartite, and
every constraint function on $U$ and $V$ comes from
two separately specified sets of constraint functions.
This is not an oversight; the reason is a serious
technical obstacle. When the graph is bipartite and, say,
$r$-regular, 
there is a curious number theoretic limitation
as to what types of subgraph fragments, called gadgets,
one can possibly construct.
It turns out that every constructible  gadget must
have a rigid arity restriction; e.g., if the gadget
represents a constraint function that can be used for
a vertices in $U$ or in $V$, the arity (the number of input variables)
of the function must be congruent to 0 modulo $r$.

We initiate in this paper 
the study of Holant problems on 
 bipartite graphs. To be specific,
we classify Holant problems on 
3-regular bipartite graphs where vertices of  one side 
 are labeled with
a nonnegative weighted
symmetric constraint function $f = [x_0, x_1, x_2, x_3]$,
which takes value $x_0, x_1, x_2, x_3$ respectively when
the input has Hamming weight 0, 1, 2, 3;
the  vertices of the other side are labeled by
the ternary equality function $(=_3)$.
Such graphs can be viewed as incidence graphs of hypergraphs.
Thus one can interpret problems of this type
as on 3-uniform (every hyperedge has size 3)
3-regular (every vertex appears in 3 hyperedges)
hypergraphs. In this view, e.g., counting
perfect matchings on 3-uniform hypergraphs 
(the number of subsets of hyperedges that cover
every vertex exactly once) corresponds to the
{\sc Exact-One} function $[0,1,0,0]$.
Alternatively one can think of them as
set systems where every set has 3 elements
and every element is in 3 subsets, and we count the number of 
exact-3-covers (\#X3C). 
Denote the problem of computing the partition function
in this case as 
$\holant{f}{(=_3)}$. We prove that
for all $f$, this problem is  \emph{either}
computable in polynomial time \emph{or}
\#P-hard, depending on an explicit
dichotomy criterion on $x_0, x_1, x_2, x_3$.
Suppose
$(X, {\cal S})$ is a 3-uniform set system where
every $x \in X$ appears in 3 sets in ${\cal S}$.
In the 0-1 case,
this dichotomy completely classifies the complexity
of  counting the number of 
ways to choose  ${\cal S'} 
\subseteq {\cal S}$ while satisfying some local constraint
specified by $f$.  In the more general nonnegative
weighted case, this is to  compute a weighted sum
of products. 

As mentioned earlier the main technical obstacle
to this regular bipartite Holant dichotomy is
the number theoretic arity restriction.
We overcome this obstacle by considering \emph{straddled}
constraint functions, i.e., those functions which have
some input variables that must be connected to one side
of the bipartite graph while some other 
 variables  must be connected to the other side.
Then we introduce a lemma that let's 
interpolate a \emph{degenerate} constraint function
by iterating the straddled function construction.
Using a Vandermonde system we can succeed in
this interpolation. Typically in proving a Holant dichotomy,
getting a degenerate constraint function is useless
and signifies failure. But here we turn the table,
and transform this ``failure'' to a ``success''
by ``peel'' off a constraint function which is a tensor factor,
whereby to break the  number theoretic arity restriction.

This paper is a mere starting point for 
understanding bipartite Holant problems. Almost
every generalization is an open problem at this point,
including more than one constraint function on either side, 
other regularity parameter $r$, real or complex valued 
constraint functions which allow cancellations, etc. 
The bigger picture
is to gain a systematic understanding of all such counting
problems in a classification program, a theme seems
second to none in its centrality to counting complexity theory,
short of proving PF $\ne$ \#P.

\section{Preliminaries}
In this paper we consider the following subclass of $\operatorname{Holant}$ problems. An input 3-regular bipartite graph $G = (U,V,E)$ is given, where each vertex on $V$ is assigned the \textsc{Equality} of arities 3 $\left(=_{3} \right)$ and each vertex on $U$ is assigned a ternary symmetric constraint function
(also called a signature) $f$ with nonnegative values. The problem is to compute the quantity
$$
\operatorname{Holant}\left(G\right) = \sum_{\sigma: E \rightarrow \{ 0,1\}} \prod_{u \in U} f\left(\sigma |_{E(u)}\right) \prod_{v \in V} \left( =_{3} \right) \left(\sigma |_{E(v)}\right)
$$

Equivalently, this  can be stated as a  weighted counting constraint satisfaction problem on Boolean variables
defined on 3-regular bipartite graphs:
The input is a 3-regular  bipartite $G = (U,V,E)$,
where  every $v \in V$ is a Boolean variable and
 every $u \in U$ represents the nonnegative valued constraint function $f$.
 An edge $(u,v) \in E$ indicates that $v$ appears in
 the constraint at $u$. Being 3-regular means that
 every constraint has 3 variables and every variable
 appears in 3 constraints.

We adopt the notation $f = [x_0,x_1,x_2,x_3]$ to represent the ternary symmetric signature $f$ where $f\left(0,0,0\right) = x_0$, $f\left(0,0,1\right)=f\left(0,1,0\right) = f\left(1,0,0\right) = x_1$, $f\left(0,1,1\right)=f\left(1,0,1\right) = f\left(1,1,0\right) = x_2$ and $f\left( 1,1,1\right) = x_3$. The \textsc{Equality} of arities 3 is $\left(=_{3} \right) = [1,0,0,1]$. For clarity, we shall call vertices in $V$ are on the right hand side (RHS) and vertices in $U$ are on the left hand side (LHS).  We denote this problem
$\holant{f}{(=_3)}$.



A gadget, such as those illustrated in Figure~\ref{gadgets},
is  a bipartite graph $G = (U, V, E_{\rm in}, E_{\rm out})$ with
internal edges $E_{\rm in}$ and dangling edges $E_{\rm out}$.
There can be $m$ dangling edges internally incident
to vertices from $U$ and $n$ dangling edges internally incident
to vertices from $V$. These $m+n$ dangling edges 
correspond to Boolean variables $x_1, \ldots, x_m, y_1, \ldots,  y_n$
and the gadget defines
a signature
\[f(x_1, \ldots, x_m, y_1, \ldots,  y_n)
=
\sum_{\sigma: E_{\rm in} \rightarrow \{ 0,1\}} \prod_{u \in U} f\left(\widehat{\sigma} |_{E(u)}\right) \prod_{v \in V} \left( =_{3} \right) \left(\widehat{\sigma}  |_{E(v)}\right),
\]
where $\widehat{\sigma} $ denotes the extension
of $\sigma$ by the assignment on the  dangling edges.

To preserve the bipartite structure, we must be careful
in any gadget construction how each  external wire is
to be connected, i.e., as an input variable
whether it is on the LHS (like those of $f$ which can be
used to  connect to $(=_3)$ on the RHS), or
it is on the RHS (like those of $(=_3)$ which can be
used to  connect to $f$ on the LHS).
See illustrations in Figure~\ref{gadgets}.

If a gadget construction produces a constraint
function $g$ such that all of its variables are on the LHS.
We claim that its arity must be a multiple of 3.
This follows from the following more general statement
that if a constraint function has $n$ input variables
on the LHS and $m$ input variables
on the RHS, then $n \equiv m \bmod 3$.
This can be easily proved by induction on the number of
occurrences of $f$ and $(=_3)$ in a  gadget construction
$\Gamma$: If either $f$ or $(=_3)$
do not occur, then the function is a tensor
product of $(=_3)$ or $(=_3)$, thus
clearly of arity $n \equiv 0 \bmod 3$ or $m \equiv 0 \bmod 3$. 
Suppose $f$ or $(=_3)$ both occur and let
$x$ be an external dangling edge. It is  internally
connected to  a vertex $v$, which is labeled
either  $f$ (if $v \in U$) or $(=_3)$ (if $v \in V$).
Let $v$  have exactly $k \in \{1, 2, 3\}$
incident edges that are dangling edges.
Now we remove $v$ and get a  gadget $\Gamma'$ with fewer
occurrences of $f$ and $(=_3)$. The induction hypothesis
and a simple accounting of the
arity for variables on the LHS and the RHS separately
complete the induction.

One  idea that is instrumental in this paper is to
use gadgets that produce \emph{straddled} signatures.
For example the gadget $G_1$ in Figure~\ref{gadgets}(a)
has one variable on the LHS and one variable on the RHS.
Such gadgets can be iterated while respecting the bipartite
structure.
%
The \emph{signature matrix} of such straddled gadgets will adopt the notation that row indices denote the input from LHS and column indices denote the input from RHS. For example, the signature matrix for $G_1$ in Figure~\ref{gadgets}(a) where we place $[x_0,x_1,x_2,x_3]$ on the square and $\left(=_3\right)$ on the circle will be $\left( \begin{array}{cc} x_0 &x_2 \\ x_1 &x_3\end{array}\right)$.  A binary signature is
\emph{degenerate} if it has determinant 0.
We will use such constructions to interpolate a {degenerate} straddled signature so that we can ``split" it to get unary signatures.
To justify that we can indeed ``split" a degenerate straddled signature, whenever we connect a unary signature, we have to ``use up" another unary signature, for the pair was created by that degenerate straddled signature. Intuitively, we can ``use all other edges up" by connecting them to form known positive global factors of the $\Holant$ value. This intuition is encoded into the following lemma.
Effectively, we can  split a degenerate binary straddled signature into unary signatures to be freely used.

\begin{lemma}\label{lm:dege-unary-split}
Let $f$ and $g$ be two nonnegative valued
signatures.
If a degenerate nonnegative binary straddled signature $\left( \begin{array}{cc} 1 &x \\ y &xy\end{array}\right)$ can be interpolated in the problem $\holant{f}{g}$,
then 
\begin{equation}\label{getting-unary}
\holant{f}{\{g, [1,x]\}} \leq_{T} \holant{f}{g}.
\end{equation}
A similar statement holds for adding the
unary $[1,y]$ on the LHS.
\end{lemma}

\begin{remark}
The same proof applies for  $\left( \begin{array}{cc} y &xy \\ 1 &x\end{array}\right)$
to get
unary signatures  $[1,x]$ on the RHS,
or unary $[y,1]$ on the LHS.
\end{remark}

\begin{proof}
We prove (\ref{getting-unary}).
Let  $f$ and $g$ have arity  $m$ and $n$ respectively.
We may assume 
that $g$ is not a
multiple of $[0,1]^{\otimes n}$ (including identically 0), for  otherwise 
$\holant{f}{\{g, [1,x]\}} $
  can be computed in PF, since
 all signatures on the RHS are degenerate and can be applied
 directed as unary signatures on copies of $f$.
 
Let  
 $k= \operatorname{gcd}\left(m,n \right)$, and $s = n/k \ge 1$.
 Consider
  any bipartite
 signature grid $\Omega = (G, \pi)$
 for  $\holant{f}{\{g, [1,x]\}}$.
 Let $N_f, N_g, N_u$ be  the numbers of occurrences of $f$,
 $g$, $[1,x]$ respectively. 
 Then we have 
 \[m N_f = n N_g + N_u,\]
 thus $N_u \equiv 0 \bmod k$.
 Let $t = N_u /k \ge 0$. We may assume
 $t \ge 1$, for otherwise $[1,x]$ does not occur and the reduction is trivial.
 
We will compute $(\operatorname{Holant}\left(G\right))^s$,  the $s$-th power of the value
$\operatorname{Holant}\left(G\right)$,
using an oracle for $\holant{f}{g}$. Since the value $\operatorname{Holant}\left(G\right)$ is nongenative,
we can obtain it from $(\operatorname{Holant}\left(G\right))^s$.

In $G$, we replace each occurrence of $[1,x]$
  with the binary
straddled signature $\left( \begin{array}{cc} 1 &x \\ y &xy\end{array}\right)$, with one end connected to  LHS,
and leaving one  edge yet to be connected to RHS. This creates  a total of
$N_u$ such edges.
This is equivalent to connecting $N_u$ copies of the  unary
signature $[1,x]$ to LHS, and having $N_u$ copies of the  unary
signature $[1,y]$ yet to be connected
 to RHS.
Now in $s$ disjoint copies of $\Omega$ there will be
$sN_u= stk$ copies of $[1,y]$ to be connected, to which we create
$t$ copies of $g$. In other words we take $g^{\otimes t}$ with total arity $stk$ and connect all $stk$ unary
signatures $[1,y]$ to it. Since $y \ge 0$ and $g$ is not a multiple of $[0,1]^{\otimes n}$,
we get an easily computable positive factor.
\end{proof}

\newpage
The main result of this paper is the following:
\begin{theorem}\label{main-thm}
$\holant{[x_0,x_1,x_2,x_3]}{(=_3)}$ where $x_i \geq 0$ for $i=0,1,2,3$ is \numP-hard except in the following cases, for which the problem is in $\operatorname{FP}$.
\begin{enumerate}
\item $[x_0,x_1,x_2,x_3]$ is degenerate;
\item $x_1 = x_2 = 0$;
\item $\left[\left(x_1 = x_3 = 0\right) \wedge \left(x_0 = x_2\right)\right]$ or $\left[\left(x_0 = x_2 = 0\right) \wedge \left(x_1 = x_3\right)\right]$.
\end{enumerate}
\end{theorem}

In case 1 the signature $[x_0,x_1,x_2,x_3]$  decomposes into
three unary signatures. In case 2 $[x_0, 0, 0, x_3]$ is a generalized
equality. In  case 3 the signature is in the affine class;
see more details about these tractable classes in~\cite{cai2017complexity}. Therefore the
$\Holant$ problem is in $\operatorname{FP}$ in cases 1--3. The main claim lies in that all other cases are \numP-hard. When proving our main result, we shall apply the following  dichotomy theorem on 2-3 $\Holant$ problem \cite{kowalczyk2010holant}:
\begin{theorem} \label{2-3}
Suppose $a,b \in \mathbb{C}$, and let $X = ab$, $Z = \left( \frac{a^3+b^3}{2} \right)^2$. Then $\holant{[a,1,b]}{\left(=_3\right)}$ is \numP-hard except in the following cases, for which the problem is in $\operatorname{P}$.
\begin{enumerate}
\item $X=1$;
\item $X=Z=0$;
\item $X=-1$ and $Z=0$;
\item $X=-1$ and $Z=-1$.
\end{enumerate}
\end{theorem} 

In fact, since this paper mainly concerns with nonnegative valued functions, when establishing \#P-hardness we usually only need to consider the exceptional case $X=1$. 

\begin{figure}[t!]
  \centering
  \begin{subfigure}[b]{0.3\linewidth}
  \begin{center}
    \includegraphics[width=0.6\linewidth]{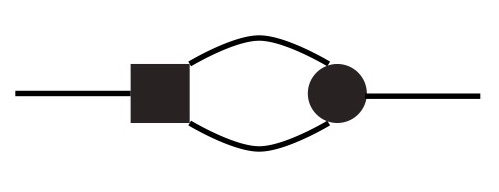}
    \caption{Binary straddled gadget $G_1$}
    \end{center}
  \end{subfigure}
  \begin{subfigure}[b]{0.3\linewidth}
    \begin{center}
    \includegraphics[width=0.6\linewidth]{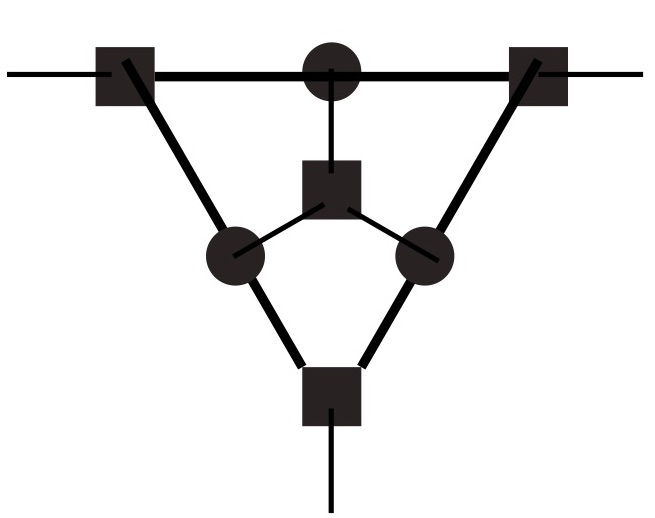}
    \caption{Ternary gadget $G_2$}
      \end{center}
  \end{subfigure}
  \\
  \begin{subfigure}[b]{0.3\linewidth}
  \begin{center}
    \includegraphics[width=0.8\linewidth]{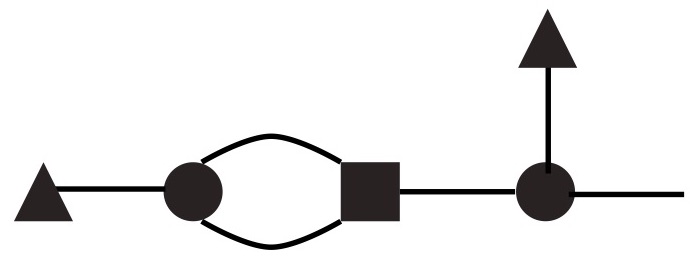}
    \caption{Unary gadget $G_3$}
      \end{center}
  \end{subfigure}
  \begin{subfigure}[b]{0.3\linewidth}
  \begin{center}
    \includegraphics[width=0.4\linewidth]{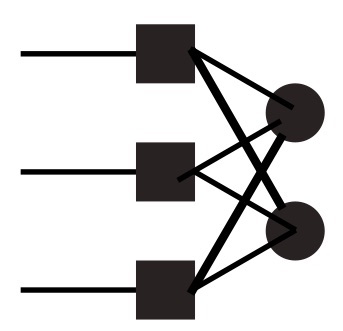}
    \caption{Ternary gadget $G_4$}
    \end{center}
  \end{subfigure}
    \caption{Some gadgets}
  \label{gadgets}
\end{figure}

\section{Interpolation From A Binary Straddled Signature}
 When 
$x_0$ and $x_3$ are not both  0, we can normalize the signature.
If $x_0 \ne 0$ we divide  the signature by $x_0$,
and get the form $[1,a,b,c]$, with $a,b,c \geq 0$.
If $x_0=0$, but $x_3 \ne 0$,
we can flip all 0 and 1 inputs, which amounts to a reversal
of the signature and get the above form.
This does not change the complexity since the Holant value is
only modified by a known nonzero factor.

Consider  arguably the simplest possible gadget $G_1$ in Figure \ref{gadgets}. We have the signature matrix $G_1 =
\left( \begin{array}{cc} 1 &b \\ a &c \end{array}\right)$. 
Let $\Delta = \sqrt{(1-c)^2+4ab}$, and assume for now 
$\Delta \ne 0$, and we take the positive square root.  Note that 
$\Delta =0$ iff $(c=1) \wedge (ab=0)$.
The matrix $G_1$  has two distinct eigenvalues
$\lambda = \frac{-\Delta+\left(1+c\right)}{2}$ and $\mu =
\frac{\Delta+\left(1+c\right)}{2}$.
If $a \neq 0$, let
$x = \frac{\Delta - \left(1-c\right)}{2a}$ and $y = \frac{\Delta + \left(1-c\right)}{2a}$. 
The matrix for $G_1$  has the Jordan Normal Form 
\begin{equation} \label{jordan form}
\left( \begin{array}{cc} 1 &b \\ a &c \end{array}\right) = 
\left( \begin{array}{cc} -x &y \\  1&1 \end{array}\right) 
\left( \begin{array}{cc} \lambda & 0 \\ 0 &\mu \end{array}\right)
 \left( \begin{array}{cc} -x &y \\  1&1 \end{array}\right)^{-1}.
\end{equation}
We now 
interpolate a binary degenerate straddled signature which will be used several times later.

\begin{lemma} \label{degenerate} 
Given the binary straddled signature $G_1 = \left( \begin{array}{cc} 1 &b \\ a &c \end{array}\right)$ with $a\neq 0$ and $\Delta > 0$, we can interpolate unary signatures $[1,x]$ on RHS or $[y,1]$ on LHS.
\end{lemma}

\begin{proof}
For $\Delta >0$,  we have $x+y = \Delta/a >0$.
Consider
$$
D = \frac{1}{x+y} \left( \begin{array}{cc} y & xy \\ 1 & x \end{array}\right) = \left( \begin{array}{cc} -x &y \\  1&1 \end{array}\right)  \left( \begin{array}{cc} 0 &0 \\  0&1\end{array}\right) \left( \begin{array}{cc} -x &y \\  1&1 \end{array}\right)^{-1}.
$$
Given any signature grid $\Omega$ where the binary degenerate straddled signature $D$ appears $n$ times, we form gadgets $G_1^s$ where $0 \le s \le n$ by iterating the $G_1$ gadget $s$ times and replacing each occurrence of $D$ with $G_1^s$. Denote the resulting
signature grid as $\Omega_{s}$. We stratify the assignments 
in the Holant sum for $\Omega$  according to assignments to
$\left( \begin{array}{cc} \lambda & 0 \\ 0 &\mu \end{array}\right)$
 as:
\begin{enumerate}
\item[-] $(0,0)$ $i$ times;
\item[-] $(1,1)$ $j$ times;
\end{enumerate}
with $i+j=n$; all  other assignments will contribute 0 in the $\operatorname{Holant}$ sum. Let $c_{i,j}$ be the sum over all such assignments of the products of evaluations (including the contributions from $\left( \begin{array}{cc} -x &y \\  1&1 \end{array}\right)$ and its inverse). Then we have
$$
\operatorname{Holant}_{\Omega_{s}} = \sum_{i+j=n} \left( \lambda^i \mu^j \right)^s \cdot c_{i,j}
$$
and
$
\operatorname{Holant}_{\Omega} = c_{0,n}.
$
Since $\Delta > 0$, the coefficients form a full rank Vandermonde matrix. Thus we can interpolate $D$ by solving the linear system of equations in polynomial time. Ignoring a nonzero factor,
we may split $D$ into unary signatures  
$[y,1]$ on LHS or $[1,x]$ on RHS.
\end{proof}

We have  $\Delta \geq |1-c|$, and $x, y \ge 0$.
Thus we can separate $D$ and obtain unary signatures 
 as specified in Lemma~\ref{lm:dege-unary-split}.
 

The following lemma lets us  interpolate unary signatures on the RHS from a binary  gadget with a straddled signature and a suitable unary signature $s$ on  the RHS. Mathematically, the proof is essentially the same as in \cite{vadhan2001complexity}, but technically Lemma~\ref{interpolate} applies to  binary straddled signatures. 

\begin{lemma} \label{interpolate}
Let $M \in \mathbb{R}^{2 \times 2}$ be the signature matrix for a binary straddled gadget which is diagonalizable with distinct eigenvalues, and let $s = [a,b]$ be a unary signature on RHS that is not a row eigenvector of $M$. Then $\{s\cdot M^{j}\}_{j \geq 0}$ can be used to interpolate any unary signature on RHS.
\end{lemma}

\section{A Basic Lemma}
In this section,  we  prove Lemma \ref{starting} which will be invoked in the remaining cases. 


When $\neg \left(x_0 = x_3 = 0\right)$, we normalize the signature $[x_0,x_1,x_2,x_3]$ to be $[1,a,b,c]$,
where $a,b,c \geq 0$. We first deal with a special case when $a=b$ and $c=1$, which will be used in the proof of  Lemma~\ref{starting}. By a slight abuse of notation, we say $[1,a,b,c]$
is \numP-hard or in FP if the 
 problem $\holant{[1,a,b,c]}{\left( =_3\right)}$ is 
\numP-hard or in FP, respectively.


\begin{lemma} \label{special}
$[1,a,a,1]$ is \numP-hard unless $a=0$ or $a=1$ in which case the problem is in $\operatorname{FP}$.
\end{lemma}

\begin{proof}
When $a=0$ or 1,  the signature $[1,a,a,1]$  is 
clearly in $\operatorname{FP}$. Suppose $a \neq 0, 1$, we have $\Delta =2a >0$. By Lemma \ref{degenerate}, we can interpolate the unary signature $[1,x] = [1,1]$ on RHS. Connect $[1,1]$ on RHS to $[1,a,a,1]$ on LHS, we get the binary signature $[1+a, 2a, 1+a]$ on LHS. Invoke Theorem \ref{2-3},  as $a \ne1$,
the problem $[1+a,2a,1+a]$ is \numP-hard, and thus $[1,a,a,1]$ is \numP-hard. 
\end{proof}

 We are now ready to prove Lemma~\ref{starting},
  the basic lemma.

\begin{lemma} \label{starting} 
When $ab \neq 0$, the problem $[1,a,b,c]$ is \numP-hard unless it is degenerate in which case the problem is in $\operatorname{FP}$.
\end{lemma}

\begin{proof}
As $ab \neq 0$, we have  $a \neq 0$, and $\Delta  >0$.
By Lemma \ref{degenerate}, we can interpolate the unary signature $[1,x]$ on RHS or $[y,1]$ on LHS where $x = \frac{\Delta - \left(1-c\right)}{2a}$ and $y = \frac{\Delta + \left(1-c\right)}{2a}$. Consider the unary gadget $G_3$
in Figure~\ref{gadgets} where we place the circles to be $\left( =_3\right)$, the triangles to be $[y,1]$, and the square to be $[1,a,b,c]$.
It has the unary signature  $[y^2+yb, ya+c]$  on the RHS. By Lemma \ref{interpolate}, we can interpolate any unary gadget, in particular $\Delta_{0} := [1, 0]$ and $\Delta_{1} := [0,1]$, on the RHS unless the  signature of $G_3$ is a row eigenvector of $G_1$. By Equation (\ref{jordan form}), the row eigenvectors of $G_1$ are the rows of
$\left( \begin{array}{cc} -x &y \\  1&1 \end{array}\right)^{-1}$,
namely proportional to $[1,-y]$ and $[1,x]$. Since $y > 0$ by $ab \ne 0$, the only exception is $\frac{ya+c}{y^2+yb} = x$.

Observe that $xy = \frac{b}{a}$. 
Solve for $y$ we get $\left(b-a^2 \right)y = \left(ac-b^2 \right)$. If $b=a^2$, then $ac = b^2$ and thus $[1,a,b,c]$ is degenerate. Otherwise, $y = \frac{ac-b^2}{b-a^2}$. Now plug into $y =  \frac{\Delta + \left(1-c\right)}{2a}$, we have $\Delta+\left(1-c\right) = \frac{2a(ac-b^2)}{b-a^2}$ which implies 
$\left(a^3-b^3-ab\left(1-c\right)\right)(ab-c) = 0$
(the high order multi-variable polynomial magically factors out).

We now divide our discussion into three cases: 
(1) $\left[ \left(a^3-b^3-ab \left(1-c\right) \neq 0 \right) \wedge \left(ab-c \neq 0\right)\right]$, 
(2) $a^3-b^3-ab\left(1-c\right) = 0$, 
and (3) $ab-c = 0$.

\textbf{Case 1:} $\left(a^3-b^3-ab\left(1-c\right) \neq 0 \right) \wedge \left(ab-c \neq 0\right)$

We can interpolate $\Delta_0$ and $\Delta_1$ on RHS. By connecting $\Delta_0$ and $\Delta_1$ on RHS to $[1,a,b,c]$ on LHS, we get the binary signatures $[1,a,b]$ and $[a,b,c]$ on LHS. That is we have  
$$
\holant{[1,a,b]}{\left(=_3\right)} \leq_T \holant{[1,a,b,c]}{\left(=_3\right)}
$$
and 
$$
\holant{[a,b,c]}{\left(=_3\right)} \leq_T \holant{[1,a,b,c]}{\left(=_3\right)}.
$$
By Theorem \ref{2-3}, the problem $\holant{[1,a,b,c]}{\left(=_3\right)}$ is \numP-hard unless $\frac{1}{a} \cdot \frac{b}{a} = 1$ and $\frac{a}{b} \cdot \frac{c}{b} = 1$, in which case
the signature $[1,a,b,c]$ is degenerate and thus the problem is in FP.

\textbf{Case 2:} $a^3-b^3-ab\left(1-c\right) = 0$

We have $1-c = \frac{a^3-b^3}{ab}$ and thus $\Delta = \sqrt{(1-c)^2+4ab} = \frac{a^3+b^3}{ab}$. Thus the unary signature interpolated by Lemma \ref{degenerate} on RHS is $[1,x] = [1,\frac{\Delta - \left(1-c\right)}{2a}] = [1, \frac{b^2}{a^2}]$. Connect $[1,x]$ on RHS to $[1,a,b,c]$ on LHS, we get the binary signature $[1+\frac{b^2}{a}, a+\frac{b^3}{a^2}, b+\frac{b^2c}{a^2}]$ on LHS. By Theorem \ref{2-3}, the problem $[1+\frac{b^2}{a}, a+\frac{b^3}{a^2}, b+\frac{b^2c}{a^2}]$ is \numP-hard and thus the problem $[1,a,b,c]$ is \numP-hard unless
$(1+\frac{b^2}{a}) (b+\frac{b^2c}{a^2})
=
({a+\frac{b^3}{a^2}})^2$,
which, after substituting $c$,
simplifies to $\left(a^2-b \right)\left(a^3+ab+2b^3 \right) = 0$. Since $a,b >0$, we have $a^3+ab+2b^3  \neq 0$ and thus $a^2-b = 0$. 
This combined with $a^3-b^3=ab\left(1-c\right)$ gives $c = a^3$, and thus $[1,a,b,c]$ is degenerate.

\textbf{Case 3:} $ab-c = 0$

Observe that under this case we have $\Delta = 1+c$. The unary signature interpolated on RHS is $[1,x] = [1,\frac{\Delta - \left(1-c\right)}{2a}]= [1,\frac{c}{a}] =[1,b]$. We connect $[1,x]$ on RHS to $[1,a,b,c]$ on LHS to get binary signature $[1+ab,a+b^2, b+bc]$ on LHS. That is we have
$$
\holant{[1+ab,a+b^2,b+bc]}{\left(=_3\right)} \leq_T \holant{[1,a,b,c]}{\left(=_3\right)}
$$
By Theorem \ref{2-3}, this problem is \numP-hard unless $(1+ab)(b+bc) = (a+b^2)^2$ which implies $\left(a^2-b \right)\left(b^3-1 \right)=0$. If $a^2-b=0$, since we are under the case $ab-c=0$, the signature $[1,a,b,c]$ is degenerate. If $b^3-1=0$, since $b$ is a positive real number, we have $b=1$. Thus the signature $[1,a,b,c]$ is simply $[1,a,1,a]$. 

Now consider the ternary gadget $G_4 = [2+2a^3, 2a+2a^2, 2a+2a^2, 2+2a^3]$ on LHS where we place the circles to be $\left( =_3\right)$ and the squares to be $[1,a,1,a]$. This signature by $G_4$ has the form in  Lemma \ref{special}, and we see that it is
 \numP-hard unless $2+2a^3 = 2a+2a^2$. So
the problem $[1,a,1,a]$ is \numP-hard and thus the problem $[1,a,b,c]$ is \numP-hard unless $a^3 -a^2 -a +1
= (a-1)^2(a+1) =0$,
which implies $a=1$. When $a=1$, the signature $[1,a,1,a]$ is degenerate.

The proof is now complete.
\end{proof}

We are now ready to prove a series of lemmas which lead to our main result, Theorem~\ref{main-thm}. The basic idea is to reduce cases to Lemma \ref{starting}, and handle exceptional cases separately. 

\section{Case $\neg (x_0 = x_3 = 0)$}
With the help of Lemma \ref{starting} we can quickly finish this
case. We may normalize $[x_0,x_1,x_2,x_3]$ to $[1,a,b,c]$
by dividing a nonzero factor, and flipping all input 0's and 1's
if necessary
(i.e., taking the reversal of the signature).
The exceptional cases in Lemma \ref{starting} for $[1,a,b,c]$
are $a=0$ or $b=0$.
We first consider the case $a \ne 0$ but $b=0$.

\begin{lemma} \label{b=0}
If $a \neq 0$ and $b = 0$, then $[1,a,b,c]$ is \numP-hard.
\end{lemma}

\begin{proof}
By gadget $G_2$ where we place the squares to be $[1,a,0,c]$ and the circles to be $\left(=_3\right)$, we have a ternary signature $[1+3a^3, a+a^4, a^2, a^3+c^4]$ on LHS. Thus Lemma \ref{starting} applies (after 
normalizing) and $[1+3a^3, a+a^4, a^2, a^3+c^4]$ is \numP-hard unless
it is degenerate,
i.e., $(a+a^4)^2 = (1+3a^3)(a^2)$ and $(a^2)^2= (a+a^4)(a^3+c^4)$.
The first equality 
forces  $a=1$
for positive $a$.
Then we get $\frac{1}{2} = 1+c^4 \ge 1$, which is a contradiction.
\end{proof}

Now consider the case  $a=0$. We have the following subcases
(1) $c=1$, (2) $c=0$, and (3) $\left(c\neq 1\right) \wedge \left(c\neq 0\right)$. These are handled 
by the following three lemmas.

\begin{lemma}
If $a = 0$ and $c = 1$, then $[1,0,b,1]$ is \numP-hard unless $b=0$, in which case the problem is in $\operatorname{FP}$.
\end{lemma}

\begin{proof}
By flipping 0 and 1 in the input, we equivalently consider the signature $[1,b,0,1]$. This case is dealt with in Lemma~\ref{b=0}
except when $b=0$ in which case it is in FP.
\end{proof}


\begin{lemma}
If $a=0$ and $c=0$, then $[1,0,b,0]$ is \numP-hard unless $b=0$ or $b=1$, in which case the problem is in $\operatorname{FP}$.
\end{lemma}

\begin{proof}
By gadget $G_1$ where we place $[1,0,b,0]$ at the squares  and
 $\left(=_3\right)$ at the circles, we have the binary degenerate straddled signature $\left( \begin{array}{cc} 1 &b \\ 0 &0 \end{array}\right)$. We can thus split and get unary signature $[1,b]$ on RHS. Connect $[1,b]$ on RHS and $[1,0,b,0]$ on LHS, we get the binary signature $[1,b^2,b]$ on LHS. By Theorem \ref{2-3}, $[1,b^2,b]$ is \numP-hard unless $b=0$ or $b=1$. When $b=0$, the original signature $[1,0,b,0]$ becomes $[1,0,0,0]$ which is degenerate and thus the problem is tractable. When $b=1$, the signature $[1,0,b,0]$ becomes $[1,0,1,0]$ which is an affine signature and thus the problem is also tractable (\cite{cai2017complexity} p.70).
\end{proof}

\begin{lemma}
If $a=0$, $c \neq 0$, and $c \neq 1$, then $[1,0,b,c]$ is \numP-hard unless $b=0$, in which case the problem is in $\operatorname{FP}$.
\end{lemma}

\begin{proof}
By flipping 0 and 1 in the input, we equivalently consider the signature $[c,b,0,1]$. Thus  Lemma \ref{b=0} applies
after normalizing, except when  $b=0$  in which case  $[c,0, 0, 1]$ is in $\operatorname{FP}$.
\end{proof}


The discussion for the case $\neg (x_0 = x_3 = 0)$ is now complete.

\section{Case $x_0=x_3=0$}
If both $x_1=x_2=0$, the signature is identically zero and the problem is trivially in FP.

Suppose exactly one of $x_1$ and $x_2$ is 0. In this case, by normalizing and possibly flipping 0 and 1 in the input, it suffices to consider the ternary signature $[0,1,0,0]$.

\begin{lemma}
The problem $\holant{[0,1,0,0]}{\left(=_3\right)}$ is \numP-hard.
\end{lemma}

\begin{proof}
We begin with a reduction from Restricted Exact Cover by 3-Set (RX3C) \cite{gonzalez1985clustering}.
This is a restricted version of  the Set Exact Cover problem
where every set has exactly 3 elements and every element is in exactly 3 sets. Given any instance for RX3C, construct the 3-regular bipartite graph whose vertices on LHS are elements of the set $X$ and vertices on RHS are 3-element subsets of $X$. Connect an edge between one vertex $v$ on LHS to one vertex $C$ on RHS if and only if $v \in C$. Then every nonzero term in the Holant sum
for  $\holant{[0,1,0,0]}{\left(=_3\right)}$ exactly
corresponds to  one solution of RX3C. We observe that the reduction given in \cite{gonzalez1985clustering} (from the Exact Cover by 3 Set
problem (X3C))  is parsimonious. And it is well known that
SAT reduces to X3C parsimoniously via 
3-Dimensional Matching (3DM).
%
Thus we have the reduction chain:
$$
\#\operatorname{SAT} \leq_{T}
  \#\operatorname{3DM} \leq_{T} \#\operatorname{X3C} \leq_{T}  \#\operatorname{RX3C} \leq_{T} \holant{[0,1,0,0]}{\left(=_3\right)}
$$

We conclude that the problem $[0,1,0,0]$ is \numP-hard.
\end{proof}

It remains to consider the case when $x_1 \cdot x_2 \neq 0$. We normalize the signature to be $[0,1,b,0]$.
\begin{lemma}
If $b \neq 0$, then  the problem $\holant{[0,1,b,0]}{\left(=_3\right)}$
 is \numP-hard.
\end{lemma}

\begin{proof}
By gadget $G_2$ where we place the squares to be $[0,1,b,0]$ and the circles to be $\left(=_3\right)$, we have the ternary signature $[3b^2, 1+2b^3, 2b+b^4, 3b^2]$ on LHS. Lemma \ref{starting} applies (after normalizing) and the problem $[3b^2, 1+2b^3, 2b+b^4, 3b^2]$ is \numP-hard unless it is degenerate. The condition for $[3b^2, 1+2b^3, 2b+b^4, 3b^2]$ being degenerate is
$$
\begin{cases}
\frac{3b^2}{1+2b^3} = \frac{1+2b^3}{2b+b^4} \\
\frac{1+2b^3}{2b+b^4} = \frac{2b+b^4}{3b^2}
\end{cases}
$$
The second equation simplifies to $(b^3-1)^2 = 0$. Thus when $b \neq 1$, we conclude that $[0,1,b,0]$ is \numP-hard. 

Proceed with the case $[0,1,1,0]$. By gadget $G_4$, where we place the squares to be $[0,1,1,0]$ and the circles to be $[1,0,0,1]$, we get the (reversal invariant) ternary signature $[3,2,2,3]$ on LHS. By Lemma \ref{starting} this problem is \numP-hard.
\end{proof}

We note that the problem $\holant{[0,1,1,0]}{\left(=_3\right)}$ 
when restricted to planar graphs is in FP.

\vspace{.1in}

\noindent
$\mathbf{Problem:}$ Pl-{\sc \#HyperGragh-Moderate-3-Cover}

\noindent
$\mathbf{Input:}$ A planar 3-uniform 3-regular hypergraph $G$.

\noindent
$\mathbf{Output:}$ The number of subsets of hyperedges
that cover every vertex with no vertex covered three times.

This is exactly the problem  $\text{Pl-Holant} ([0,1,1,0]\mid \left(=_3\right))$, the restriction of $\holant{[0,1,1,0]}{\left(=_3\right)}$ 
 to planar graphs. 
Its P-time tractability is  seen by the following holographic
reduction using the Hadamard matrix
$H = 
\left( \begin{array}{cc} 1 &1 \\ 1 &-1 \end{array}\right)
$ to counting weighted perfect matchings.
Under this transformation both signatures
$[0,1,1,0]$ and $(=_3)$ are transformed to matchgate
signatures.
\[H^{\otimes 3}  (=_3)
=
H^{\otimes 3} 
\left[
\left( \begin{array}{cc} 1 \\ 0 \end{array}\right)^{\otimes 3}
+
\left( \begin{array}{cc} 0 \\ 1 \end{array}\right)^{\otimes 3}
\right] 
=
\left( \begin{array}{cc} 1 \\ 1 \end{array}\right)^{\otimes 3}
+
\left( \begin{array}{cc} 1 \\ -1 \end{array}\right)^{\otimes 3}
=
[2,0,2,0],
\]
and 
\[
[0,1,1,0] (H^{-1})^{\otimes 3}
=
\frac{1}{8}
\left[
(1,~~1)^{\otimes 3} - (1,~~0)^{\otimes 3} - (0,~~1)^{\otimes 3}
\right]
H^{\otimes 3}
=
\frac{1}{4}
[3, 0, -1, 0].
\]
Since both these are matchgate signatures,
the problem is reduced to counting weighted perfect matchings on planar graphs.
Thus the planar problem 
$\text{Pl-}\holant{[0,1,1,0]}{\left(=_3\right)}$ 
can be computed in polynomial time using Kasteleyn's
algorithm (see \cite{valiant2008holographic} and \cite{cai2017complexity}). For readers unfamiliar with
holographic algorithms, we will describe this algorithm
for Pl-{\sc \#HyperGragh-Moderate-3-Cover} in more detail
in an Appendix.

\section*{Acknowledgement}
We sincerely thank Shuai Shao for many helpful discussions.

\section*{Appendix}
The holographic algorithm
for Pl-{\sc \#HyperGragh-Moderate-3-Cover} is achieved by a
many-to-many reduction from this problem to
 the problem of counting (weighted) planar perfect matchings
 (Pl-\#{\sc PM}).
 Given planar weighted graph $G = (V, E, w)$
 where $w : E \rightarrow \mathbb{R}$ is a weight function
 for the edge set, any perfect matching $M \subseteq E$ 
 has weight $w(M) = \prod_{e \in M} w(e)$, and the partition function  of
 the problem  Pl-\#{\sc PM} is
\[\#{\rm PM} (G) = \sum_{{\rm perfect~matchings}~M} w(M).\]
Notice that when all weights $w(e)=1$, i.e., in the unweighted
case, $\#{\rm PM} (G)$ 
simply counts the number of perfect matchings in $G$.
Kasteleyn's
algorithm can compute this  partition function $\#{\rm PM} (G)$
in polynomial time on planar graphs.  It is important to note that
Kasteleyn's
algorithm can handle the case when $w(e)$ are both positive and
negative real 
(even complex) numbers, which is important in this case for us.

The values of $\holant{[0,1,1,0]}{\left(=_3\right)}$
and
$\holant{\frac{1}{4}[3,0,-1,0]}{[2,0,2,0]}$
on the same instance graph $G$
are exactly the same by Valiant's Holant Theorem~\cite{{valiant2008holographic}}, as
\[
[0,1,1,0] (H^{-1})^{\otimes 3}
=
\frac{1}{4}
[3, 0, -1, 0]
~~~~\mbox{and}~~~~
H^{\otimes 3}  (=_3)
=
[2,0,2,0].
\]
Notice that this is a many-to-many transformation
among solutions of one problem  $\holant{[0,1,1,0]}{\left(=_3\right)}$
(subsets of hyperedges
that cover every vertex with no vertex covered three times)
and another $\holant{\frac{1}{4}[3,0,-1,0]}{[2,0,2,0]}$.

Next we will use the following gadgets, called matchgates,
to ``implement'' the constraint functions
$\frac{1}{4}[3,0,-1,0]$ and $[2,0,2,0]$
by perfect matchings.  Consider the following two
matchgates.

\begin{figure}[h!]
  \centering
  \begin{subfigure}[b]{0.3\linewidth}
  \begin{center}
    \includegraphics[width=0.6\linewidth]{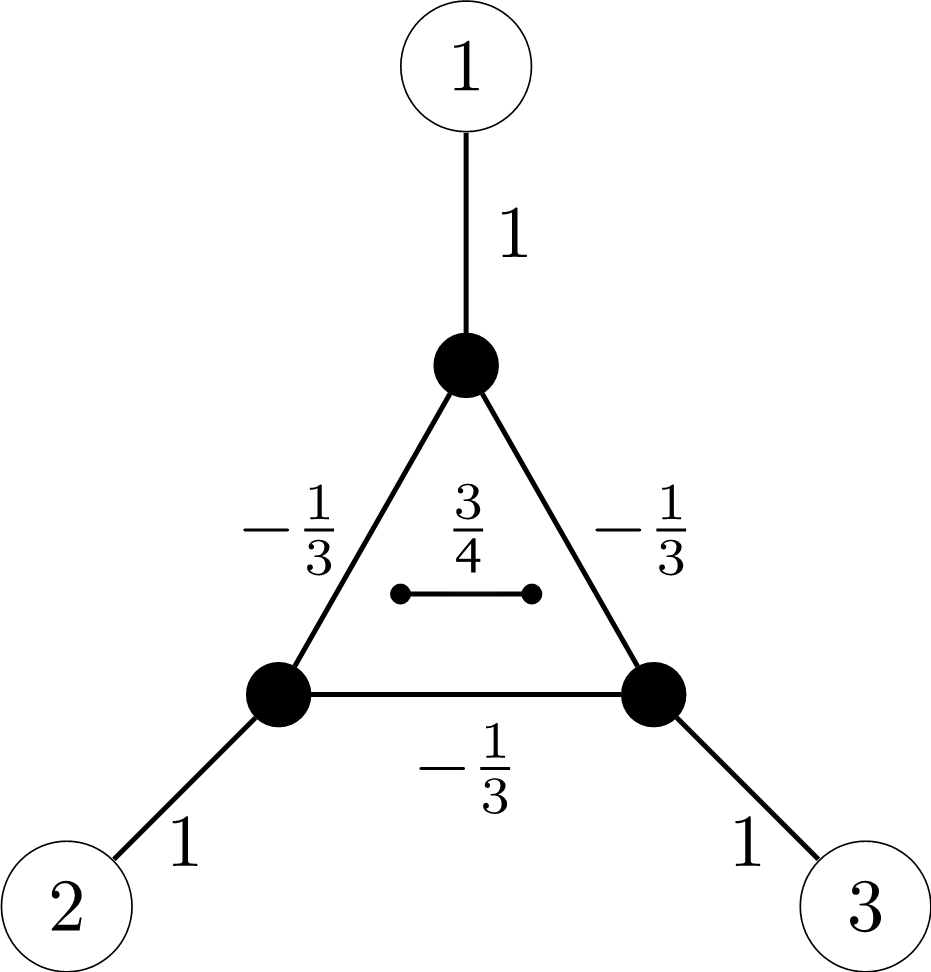}
    \caption{Matchgate for $\frac{1}{4}[3,0,-1,0]$}
    \end{center}
  \end{subfigure}
  \begin{subfigure}[b]{0.3\linewidth}
    \begin{center}
    \includegraphics[width=0.6\linewidth]{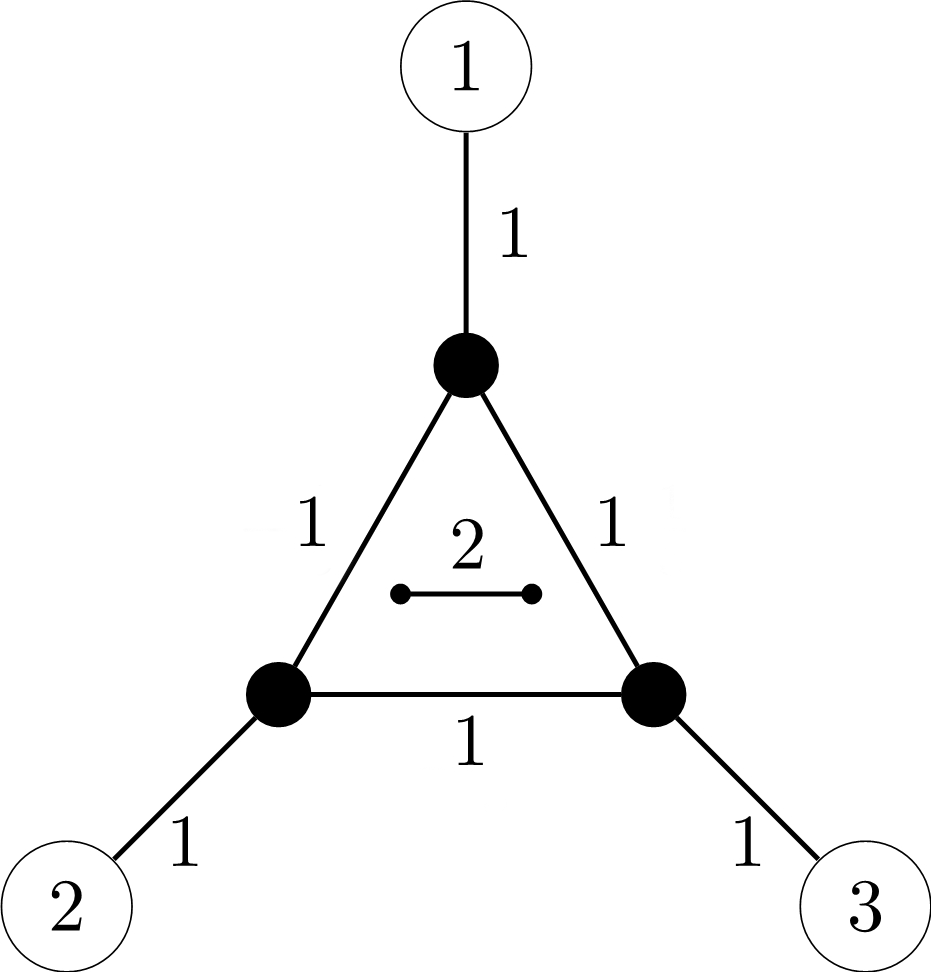}
    \caption{Matchgate for $[2,0,2,0]$}
      \end{center}
  \end{subfigure}
  \caption{Matchgates}
  \label{machgates-figures}
\end{figure}

For the matchgate in Figure~\ref{machgates-figures} (a),
 we consider its perfect matchings, when any subset
 $\emptyset \subseteq S \subseteq \{1,2,3\}$
 of the external nodes labeled $1,2,3$ are removed.
 For $S = \emptyset$ there is a unique perfect matching
 $M$ with weight $w(M) = 3/4$.
 If $|S| = 2$, we get $w(S) = -1/3$.
 If $|S|$ is odd, then $w(S) = 0$.
 In this sense the matchgate has the symmetric
 signature $\frac{1}{4}[3,0,-1,0]$.
 
A similar calculation shows that the  the matchgate
in Figure~\ref{machgates-figures} (b), has the symmetric
 signature $[2,0,2,0]$.
 
 Now we take a planar input bipartite graph $G$
 for  $\text{Pl-Holant} ([0,1,1,0]\mid \left(=_3\right))$,
 replace each vertex of degree three on the LHS
 (which has the label $[0,1,1,0]$)
 by the matchgate in Figure~\ref{machgates-figures} (a),
 replace each vertex of degree three on the RHS
 (which has the label $[1,0,0,1]$)
 by the matchgate in Figure~\ref{machgates-figures} (b),
 and for any edge in $G$
 add one edge with weight 1 between the corresponding
 external vertices from their respective matchgates.
 This creates a planar graph $G'$, of size linearly bounded by that  of
 $G$.
 
 A moment reflection should convince the reader that
 the value $\#{\rm PM} (G')$ is exactly the same
 as the Holant value 
 $\operatorname{Holant}\left(G\right)$.

\bibliographystyle{plain}
\bibliography{References}
 

\end{document}